%% file: UniGen.tex
\documentclass{sig-alternate}

\usepackage{amsthm}
\usepackage{enumerate}
\usepackage{graphicx} 
\usepackage{graphics, color}
\usepackage{algorithm}
\usepackage{algpseudocode}
\usepackage{boxedminipage} 
\usepackage{verbatim}
\usepackage[bookmarks=false]{hyperref}
\usepackage{array}
\usepackage{multirow}
\usepackage{color}
\usepackage{changebar}
\usepackage{multicol}

\newcommand{\XORSampleprime}{\ensuremath{\mathsf{XORSample}'}}
\newcommand{\BoundedSAT}{\ensuremath{\mathsf{BSAT}}}

\newcommand{\UniformWitness}{\ensuremath{\mathsf{UniWit}}}
\newcommand{\UniGen}{\ensuremath{\mathsf{UniGen}}}
\newcommand{\PAWS}{\ensuremath{\mathsf{PAWS}}}

\newcommand{\killthis}[1]{}
\newcommand{\prob}{\ensuremath{\mathsf{Pr}}}
\newcommand{\expect}{\ensuremath{\mathsf{E}}}
\newcommand{\var}{\ensuremath{\mathsf{V}}}

\newcommand{\NP}{\ensuremath{\mathsf{NP}}}

\newcommand{\SAT}{\ensuremath{\mathsf{SAT}}}
\newcommand{\sharpSATTool}{\ensuremath{\mathsf{sharpSAT}}}

\newcommand{\approxCounter}{\ensuremath{\mathsf{ApproxModelCounter}}}
\newcommand{\approxMC}{\ensuremath{\mathsf{ApproxMC}}}

\newcommand{\ComputeKappaAndPivot}{\ensuremath{\mathsf{ComputeKappaPivot}}}

\newcommand{\US}{\ensuremath{\mathsf{US}}}

\newcommand{\UniWit}{\ensuremath{\mathsf{UniWit}}}
\newtheoremstyle{newstyle}      
  {.5\baselineskip\@plus.2\baselineskip\@minus.2\baselineskip}%
  {.5\baselineskip\@plus.2\baselineskip\@minus.2\baselineskip}%
{\mdseries} %
{} %
{\bfseries} %
{.} %
{ } %
{} %

\newtheorem{lemma}{Lemma}
\newtheorem{theorem}{Theorem}
\algrenewcomment[1]{\(\triangleright\) #1}

\algnotext{EndFor}
\algnotext{EndIf}
\begin{document}
\conferenceinfo{DAC}{'14, June 01-05 2014, San Francisco, CA, USA}		
\title{Balancing Scalability and Uniformity in SAT Witness Generator
\thanks{\scriptsize This work was supported in part by NSF grants CNS 1049862 and CCF-1139011,
by NSF Expeditions in Computing project "ExCAPE: Expeditions in Computer
Augmented Program Engineering", by BSF grant 9800096, by gift from
Intel, by a grant from Board of Research in Nuclear Sciences, India, and by the Shared University Grid at Rice funded by NSF under Grant EIA-0216467, and a partnership between Rice University, Sun Microsystems, and Sigma Solutions, Inc. }}

\numberofauthors{2} 

\author{
\alignauthor
Supratik Chakraborty\\
       \affaddr{Indian Institute of Technology, Bombay}\\
       \email{supratik@cse.iitb.ac.in}
\alignauthor
Kuldeep S. Meel, Moshe Y. Vardi\\
       \affaddr{Rice University}\\
       \email{kuldeep@rice.edu,vardi@cs.rice.edu}
}

\maketitle
\input{Abstract}
\input{Introduction}
\input{Prelims}
\input{RelatedWork}
\input{Algorithm}

\input{ExperimentalMethodology}
\input{Conclusion}

 \bibliography{Report}
\bibliographystyle{abbrv}

\clearpage
\appendix

\input{Appendix}
\end{document}

%% file: Abstract.tex
\section*{ABSTRACT}
Constrained-random simulation is the predominant approach used in the industry for functional verification of complex digital designs.  The effectiveness of this approach depends on two key factors: the quality of constraints used to generate test vectors, and the randomness of solutions generated from a given set of constraints.  In this paper, we focus on the second problem, and present an algorithm that significantly improves the state-of-the-art of (almost-)uniform generation of solutions of large Boolean constraints. Our algorithm provides strong theoretical guarantees on the uniformity of generated solutions and scales to problems involving hundreds of thousands of variables.

%% file: Introduction.tex
\section{Introduction}
Functional verification constitutes one of the most challenging and
time-consuming steps in the design of modern digital systems.  The
primary objective of functional verification is to expose design bugs
early in the design cycle.  Among various techniques available for
this purpose, those based on simulation overwhelmingly dominate
industrial practice.  In a typical simulation-based functional
verification exercise, a gate-level or RTL model of the circuit is
simulated for a large number of cycles with specific input patterns.
The values at observable outputs, as computed by the simulator, are
then compared against their expected values, and any discrepancy is
flagged as manifestaton of a bug.  The state of simulation technology
today is mature enough to allow simulation of large designs within
reasonable time using modest computational resources.  Generating
input patterns that exercise diverse corners of the
design's behavior space, however, remains a challenging problem
\cite{BF01}.

In recent years, constrained-random simulation (also called
constrained-random verification, or CRV) \cite{NavRim07} has emerged
as a practical approach to address the problem of simulating designs
with ``random enough'' input patterns.  In CRV, the verification
engineer declaratively specifies a set of constraints on the values of
circuit inputs.  Typically, these constraints are obtained from usage
requirements, environmental constraints, constraints on operating
conditions and the like.  A constraint solver is then used to generate
random values for the circuit inputs satisfying the constraints.
Since the distribution of errors in the design's behavior space is not
known \emph{a priori}, every solution to the set of constraints is as
likely to discover a bug as any other solution.  It is therefore
important to sample the space of all solutions uniformly or
almost-uniformly (defined formally below) at random.  Unfortunately,
guaranteeing uniformity poses significant technical challenges when
scaling to large problem sizes.  This has been repeatedly noted in the
literature (see, for example, \cite{DKBZ09,PMB08,Kitchen2010markov})
and also confirmed by industry practitioners\footnote{Private
  communication: R. Kurshan}.  %
The difficulties of generating solutions with guarantees of uniformity
have even prompted researchers to propose alternative techniques for
generating input patterns~\cite{DKBZ09,PMB08}.  This paper takes a
step towards remedying this situation.  Specifically, we describe an
algorithm for generating solutions to a set of Boolean
constraints, with stronger guarantees on uniformity and with higher
scalability in practice than that achieved earlier.

Since constraints that arise in CRV of digital circuits are encodable
as Boolean formulae, we focus on uniform generation of solutions of
Boolean formulae.  Henceforth, we call such solutions \emph{{\SAT}
  witnesses}.  Besides its usefulness in CRV and in other
applications~\cite{Bacchus2003,Roth1996}, uniform generation of SAT
witnesses has had strong theoretical interest as well~\cite{Jerr}.
Most prior approaches to solving this problem belong to one of two
categories: those that focus on strong guarantees of uniformity but
scale poorly in practice (examples
being~\cite{Yuan2004,Bellare98uniformgeneration,Jerr}),
and those that provide practical heuristics to scale to large problem
instances with weak or no guarantees of uniformity (examples
being~\cite{dechter2002,Kitchen2010markov,Selman-Sampling})).
In~\cite{SKV13}, Chakraborty, Meel and Vardi attempted to bridge these
extremes through an algorithm called {\UniWit}.  More recently, Ermon,
Gomes, Sabharwal and Selman~\cite{EGSS13} proposed an algorithm called
{\PAWS} for sampling witnesses from discrete distributions over large
dimensions.  While {\PAWS} is designed to work with any discrete
distribution specified through a graphical model, for purposes of this
paper, we focus only on distributions that assign equal weight to
every assignment.  For such distributions, both {\PAWS} and {\UniWit}
represent alternative (albeit related) approaches to solve the same
problem -- that of uniform generation of SAT witnesses.
Unfortunately, both algorithms suffer from inherent limitations that
make it difficult to scale them to Boolean constraints with tens of
thousands of variables and beyond.  In addition, the guarantees
provided by these algorithms (in the context of uniform generation of
SAT witnesses) are weaker than what one would desire in practice.

In this paper, we propose an algorithm called {\UniGen} that addresses
some of the deficiencies of {\UniWit} and {\PAWS}.  This enables us to
improve both the theoretical guarantees \emph{and} practical
performance vis-a-vis earlier algorithms in the context of uniform
generation of SAT witnesses.  {\UniGen} is the first algorithm to
provide strong two-sided guarantees of almost-uniformity, while
scaling to problems involving hundreds of thousands of variables.  We
also improve upon the success probability of the earlier algorithms
significantly, both in theory and as evidenced by our experiments.

%% file: Prelims.tex
\section{Notation and Preliminaries}\label{sec:prelims}

Let $F$ be a Boolean formula in conjunctive normal form (CNF), and let
$X$ be the set of variables appearing in $F$.  The set $X$ is called
the \emph{support} of $F$.  A \emph{satisfying assignment} or
\emph{witness} of $F$ is an assignment of truth values to variables in
its support such that $F$ evaluates to true.  We denote the set of all
witnesses of $F$ as $R_F$.  Let ${\mathcal D} \subseteq X$ be a subset
of the support such that there are no two satisfying assignments of
$F$ that differ only in the truth values of variables in ${\mathcal
  D}$.  In other words, in every satisfying assignment of $F$, the
truth values of variables in $X\setminus {\mathcal D}$ uniquely
determine the truth value of every variable in ${\mathcal D}$.  The
set ${\mathcal D}$ is called a \emph{dependent} support of $F$, and
$X\setminus \mathcal{D}$ is called an \emph{independent} support of
$F$.  Note that there may be more than one independent supports of $F$.
For example, $(a \vee \neg b) \wedge (\neg a \vee b)$ has three
independent supports: $\{a\}$, $\{b\}$ and $\{a, b\}$. Clearly, if
${\mathcal{I}}$ is an independent support of $F$, so is every superset
of ${\mathcal{I}}$.  For notational convenience, whenever the formula
$F$ is clear from the context, we omit mentioning it.

We use $\prob\left[X: {\cal P} \right]$ to denote the probability of
outcome $X$ when sampling from a probability space ${\cal P}$.  For
notational clarity, we omit ${\cal P}$ when it is clear from the
context.  The expected value of the outcome $X$ is denoted
$\expect\left[X\right]$.  Given a Boolean formula $F$, a
\emph{probabilistic generator} of witnesses of $F$ is a probabilistic
algorithm that generates a random witness in $R_F$.  A \emph{uniform
  generator} $\mathcal{G}^{u}(\cdot)$ is a probabilistic generator
that guarantees $\prob\left[\mathcal{G}^{u}(F) = y\right] = 1/|R_F|$,
for every $y \in R_F$.  An \emph{almost-uniform generator}
$\mathcal{G}^{au}(\cdot, \cdot)$ ensures that for every $y \in R_F$,
we have $\frac{1}{(1+ \varepsilon)|R_F|}$ $\le$
$\prob\left[\mathcal{G}^{au}(F,\varepsilon) = y\right]\le$ $\frac{1 +
  \varepsilon}{|R_F|}$, where $\varepsilon > 0$ is the specified
\emph{tolerance}.  A \emph{near-uniform generator}
$\mathcal{G}^{nu}(\cdot)$ further relaxes the guarantee of uniformity,
and ensures that $\prob\left[\mathcal{G}^{nu}(F) = y\right] \geq
c/|R_F|$ for a constant $c$, where $0 < c \leq 1$.  Probabilistic
generators are allowed to occasionally ``fail" in the sense that no
witness may be returned even if $R_F$ is non-empty.  The failure
probability for such generators must be bounded by a constant strictly
less than 1.  The algorithm presented in this paper falls in the
category of almost-uniform generators.  An idea closely related to
that of almost-uniform generation, and used in a key manner in our
algorithm, is \emph{approximate model counting}.  Given a CNF formula
$F$, an \emph{exact model counter} returns the size of $R_F$.  An
\emph{approximate model counter} ${\approxMC}(\cdot, \cdot, \cdot)$
relaxes this requirement to some extent.  Given a CNF formula $F$, a
tolerance $\varepsilon > 0$ and a confidence $1-\delta \in (0, 1]$,
  and approximate model counter ensures that
  $\prob[\frac{|R_F|}{1+\varepsilon} \le {\approxMC}(F, \varepsilon,
    1-\delta) \le (1+\varepsilon)|R_F|] \ge 1-\delta$.
 
A special class of hash functions, called
\emph{$r$-wise independent} hash functions, play a crucial role in
our work.  Let $n, m$ and $r$ be positive integers, and let
$H(n,m,r)$ denote a family of $r$-wise independent hash functions
mapping $\{0, 1\}^n$ to $\{0, 1\}^m$.  We use $h \xleftarrow{R}
H(n,m,r)$ to denote the probability space obtained by choosing a
hash function $h$ uniformly at random from $H(n,m,r)$.  The property
of $r$-wise independence guarantees that for all $\alpha_1, \ldots
\alpha_r \in \{0,1\}^m $ and for all distinct $y_1, \ldots y_r \in
\{0,1\}^n$, $\prob\left[\bigwedge_{i=1}^r h(y_i) = \alpha_i\right.$
$\left.: h \xleftarrow{R} H(n, m, r)\right] = 2^{-mr}$.  For every
$\alpha \in \{0, 1\}^m$ and $h \in H(n, m, r)$, let $h^{-1}(\alpha)$
denote the set $\{y \in \{0, 1\}^n \mid h(y) = \alpha\}$.  Given
$R_F \subseteq \{0, 1\}^n$ and $h \in H(n, m, r)$, we use $R_{F, h,
\alpha}$ to denote the set $R_F \cap h^{-1}(\alpha)$.  If we keep
$h$ fixed and let $\alpha$ range over $\{0, 1\}^m$, the sets $R_{F,
h, \alpha}$ form a partition of $R_F$. 

%% file: RelatedWork.tex
\section{Related Work}\label{sec:related_work}

Marrying scalability with strong guarantees of uniformity has been the
holy grail of algorithms that sample from solutions of constraint
systems.  The literature bears testimony to the significant tension
between these objectives when designing random generators of SAT
witnesses.  Earlier work in this area either provide strong
theoretical guarantees at the cost of scalability, or remedy the
scalability problem at the cost of guarantees of
uniformity.  More recently, however, there have been efforts to bridge
these two extremes.

Bellare, Goldreich and Petrank~\cite{Bellare98uniformgeneration}
showed that a provably uniform generator of SAT witnesses can be
designed in theory to run in probabilistic polynomial time relative to
an {\NP} oracle.  Unfortunately, it was shown in~\cite{SKV13} that
this algorithm does not scale beyond formulae with few tens of
variables in practice.  Weighted binary decision diagrams (BDD) have
been used in~\cite{Yuan2004} to sample uniformly from SAT witnesses.
However, BDD-based techniques are known to suffer from scalability
problems~\cite{Kitchen2010markov}.  Adapted BDD-based techniques with
improved performance were proposed in~\cite{kukula2000}; however, the
scalability was achieved at the cost of guarantees of uniformity.
Random seeding of DPLL SAT solvers~\cite{chaff2001} has been shown to
offer performance, although the generated distributions of witnesses
can be highly skewed~\cite{Kitchen2010markov}.

Markov Chain Monte Carlo methods (also called MCMC
methods)~\cite{Kitchen2010markov,wei2005new} are widely considered to
be a practical way to sample from a distribution of solutions.
Several MCMC algorithms, such as those based on simulated annealing,
Metropolis-Hastings algorithm and the like, have been studied
extensively in the literature~\cite{Kirkpatrick83,madras2002}.  While
MCMC methods guarantee eventual convergence to a target distribution
under mild requirements, convergence is often impractically slow in
practice.  
The work
of~\cite{wei2005new,Kitchen2010markov} proposed several such
adaptations for MCMC-based sampling in the context of
constrained-random verification.  Unfortunately, most of these
adaptations are heuristic in nature, and do not preserve theoretical
guarantees of uniformity.  
constraints, thereby increasing constraint-solving time.  
Sampling
techniques based on interval-propagation and belief networks have been
proposed in~\cite{dechter2002,gogate2006,iyer2003race}.  The
simplicity of these approaches lend scalability to the techniques, but
the generated distributions can deviate significantly from the uniform
distribution, as shown in~\cite{KitKue2007}.

Sampling techniques based on hashing were originally pioneered by
Sipser~\cite{Sipser83}, and have been used subsequently by several
researchers~\cite{Bellare98uniformgeneration,Gomes-Sampling,SKV13}.
The core idea in hashing-based sampling is to use $r$-wise independent
hash functions (for a suitable value of $r$) to randomly partition the
space of witnesses into ``small cells" of roughly equal size, and then
randomly pick a solution from a randomly chosen cell.  The algorithm
of Bellare et al.~referred to above uses this idea with $n$-wise
independent algebraic hash functions (where $n$ denotes the size of
the support of $F$).  As noted above, their algorithm scales very
poorly in practice.  Gomes, Sabharwal and Selman used $3$-wise
independent linear hash functions in~\cite{Gomes-Sampling} to design
{\XORSampleprime}, a near-uniform generator of SAT witnesses.
Nevertheless, to realize the guarantee of near-uniformity, their
algorithm requires the user to provide difficult-to-estimate input
parameters.  Although {\XORSampleprime} has been shown to scale to
constraints involving a few thousand variables, Gomes et
al. acknowledge the difficulty of scaling their algorithm to much
larger problem sizes without sacrificing theoretical
guarantees~\cite{Gomes-Sampling}.

Recently, Chakraborty, Meel and Vardi~\cite{SKV13} proposed a new
hashing-based SAT witness generator, called {\UniformWitness}, that
represents a small but significant step towards marrying the
conflicting goals of scalability and guarantees of uniformity.  Like
{\XORSampleprime}, the {\UniformWitness} algorithm uses $3$-wise
independent linear hashing functions.  Unlike {\XORSampleprime},
however, the guarantee of near-uniformity of witnesses generated by
{\UniWit} does not depend on difficult-to-estimate input parameters.
In~\cite{SKV13}, {\UniformWitness} has been shown to scale to formulas
with several thousand variables.  In addition, Chakraborty et al
proposed a heuristic called ``leap-frogging'' that allows
{\UniformWitness} to scale even further -- to tens of thousands of
variables~\cite{SKV13}.  Unfortunately, the guarantees of
near-uniformity can no longer be established for {\UniformWitness}
with ``leap-frogging''.  More recently, Ermon et al.~\cite{EGSS13} 
proposed a hashing-based algorithm called {\PAWS}
for sampling from a distribution defined over a discrete set using a
graphical model.  While the algorithm presented in this paper has some
similarities with {\PAWS}, there are significant differences as well.
Specifically, our algorithm provides much stronger theoretical
guarantees vis-a-vis those offered by {\PAWS} in the context of
uniform generation of SAT witness.  In addition, our algorithm scales
to hundreds of thousands of variables while preserving the theoretical
guarantees.  {\PAWS} faces the same scalability hurdles as
{\UniformWitness}, and is unlikely to scale beyond a few thousand
variables without heuristic adapatations that compromise its
guarantees.

%% file: Algorithm.tex
\section{The UniGen Algorithm}\label{sec:algorithm}
The new algorithm, called {\UniGen}, falls in the category of
hashing-based almost-uniform generators.  {\UniGen}
shares some features with earlier hashing-based algorithms such as
{\XORSampleprime}~\cite{Gomes-Sampling},
{\UniformWitness}~\cite{SKV13} and {\PAWS}~\cite{EGSS13},
but there are key differences that allow {\UniGen} to
significantly outperform these earlier algorithms, both in terms of
theoretical guarantees and measured performance.

Given a CNF formula $F$, we use a family of $3$-independent
hash functions to randomly partition the set, $R_F$, of witnesses of
$F$.  Let $h: \{0, 1\}^n \rightarrow \{0, 1\}^m$ be a hash function in
the family, and let $y$ be a vector in $\{0, 1\}^n$.  Let $h(y)[i]$
denote the $i^{th}$ component of the vector obtained by applying $h$
to $y$.  The family of hash functions of interest is defined as
$\{h(y) \mid h(y)[i] = a_{i,0} \oplus (\bigoplus_{k=1}^n a_{i,k}\cdot
y[k]), a_{i,j} \in \{0, 1\}, 1 \leq i \le m, 0 \leq j \leq n\}$, where
$\oplus$ denotes the xor operation.  By choosing values of $a_{i,j}$
randomly and independently, we can effectively choose a random hash
function from the family.  It has been shown in~\cite{Gomes-Sampling}
that this family of hash functions is $3$-independent.  Following
notation introduced in Section~\ref{sec:prelims}, we call this family
$H_{xor}(n, m, 3)$.

While $H_{xor}(n, m, 3)$ was used earlier in {\XORSampleprime},
{\PAWS}, and (in a variant of) {\UniformWitness}, there is a
fundamental difference in the way we use it in {\UniGen}.  Let $X =
\{x_1, x_2, \ldots x_{|X|}\}$ be the set of variables of $F$.  Given
$m > 0$, the algorithms {\XORSampleprime}, {\PAWS} and {\UniformWitness}
partition $R_F$ by randomly choosing $h \in H_{xor}(|X|, m, 3)$ and
$\alpha \in \{0, 1\}^m$, and by seeking witnesses of $F$ conjoined
with $\bigwedge_{i=1}^m\left(h(x_1, \ldots x_{|X|})[i] \leftrightarrow
\alpha[i]\right)$.  By choosing a random $h(x_1, \ldots x_{|X|}) \in
H_{xor}(|X|, m, 3)$, the set of \emph{all} assignments to variables in
$X$ (regardless of whether they are witnesses of $F$) is partitioned
randomly.  This, in turn, ensures that the set of satisfying
assignments of $F$ is also partitioned randomly.  Each conjunctive
constraint of the form $(h(x_1 \ldots x_{|X|})[i]$ $\leftrightarrow$
$\alpha[i])$ is an xor of a subset of variables of $X$ and
$\alpha[i]$, and is called an \emph{xor-clause}.  Observe that the
expected number of variables in each such xor-clause is approximately
$|X|/2$.  It is well-known (see, for example~\cite{ghss07:shortxors})
that the difficulty of checking satisfiability of a CNF formula with
xor-clauses grows significantly with the number of variables per
xor-clause.  It is therefore extremely difficult to scale
{\XORSampleprime}, {\PAWS} or {\UniformWitness} to problems involving
hundreds of thousands of variables.  In~\cite{SKV13}, an alternative
family of linear hash functions is proposed to be used with
{\UniformWitness}.  Unfortunately, this also uses $|X|/2$ variables
per xor-clause on average, and suffers from the same problem.
In~\cite{ghss07:shortxors}, a variant of $H_{xor}(|X|, m, 3)$ is used,
wherein each variable in $X$ is chosen to be in an xor-clause with a
small probability $q$ ($< 0.5$).  This mitigates the performace
bottleneck significantly, but theoretical guarantees of (near-)uniformity 
are lost.

We address the above problem in {\UniGen} by making two important
observations: (i) an independent support ${\mathcal{I}}$ of $F$ is
often far smaller (sometimes by a few orders of magnitude) than $X$,
and (ii) since the value of every variable in $X\setminus
{\mathcal{I}}$ in a satisfying assignment of $F$ is uniquely
determined by the values of variables in ${\mathcal{I}}$, the set
$R_F$ can be randomly partitioned by randomly partitioning its
projection on ${\mathcal{I}}$.  This motivates us to design an
almost-uniform generator that accepts a subset $S$ of the support of
$F$ as an additional input.  We call $S$ the set of \emph{sampling
variables} of $F$, and intend to use an independent support of $F$
(not necessarily a minimal one) as the value of $S$ in any invocation
of the generator.  Without loss of generality, let $S = \{x_1, \ldots
x_{|S|}\}$, where $|S| \leq |X|$.  The set $R_F$ can now be
partitioned by randomly choosing $h \in H_{xor}(|S|, m, 3)$ and
$\alpha \in \{0, 1\}^m$, and by seeking solutions of $F \wedge
\bigwedge_{i=1}^m\left(h(x_1, \ldots x_{|S|})[i] \leftrightarrow
\alpha[i]\right)$.  If $|S| \ll |X|$ (as is often the case in our
experience), the expected number of variables per xor-clause is
significantly reduced. This makes satisfiability checking easier, and
allows scaling to much larger problem sizes than otherwise possible.
It is natural to ask if finding an independent support of a CNF
formula $F$ is computationally easy.  While an algorithmic solution to
this problem is beyond the scope of this paper, our experience
indicates that a small, not necessarily minimal, independent support
can often be easily determined from the source domain from which the
CNF formula $F$ is derived.  For example, when a non-CNF formula $G$
is converted to an equisatisfiable CNF formula $F$ using Tseitin
encoding, the variables introduced by the encoding form a dependent
support of $F$.

The effectiveness of a hashing-based probabilistic generator depends
on its ability to quickly partition the set $R_F$ into ``small'' and
``roughly equal'' sized random cells.  This, in turn, depends on the
parameter $m$ used in the choice of the hash function family $H(n, m,
r)$.  A high value of $m$ leads to skewed distributions of sizes of
cells, while a low value of $m$ leads to cells that are not small
enough.  The best choice of $m$ depends on $|R_F|$, which is not known
\emph{a priori}.  Different algorithms therefore use different
techniques to estimate a value of $m$.  In {\XORSampleprime}, this is
achieved by requiring the user to provide some difficult-to-estimate
input parameters.  In {\UniformWitness}, the algorithm sequentially
iterates over values of $m$ until a good enough value is found.  The
approach of {\PAWS} comes closest to our, although there are crucial
differences.  In both {\PAWS} and {\UniGen}, an approximate model
counter is first used to estimate $|R_F|$ within a specified tolerance
and with a specified confidence.  This estimate, along with a
user-provided parameter, is then used to determine a \emph{unique}
value of $m$ in {\PAWS}.  Unfortunately, this does not facilitate
proving that {\PAWS} is an almost-uniform generator.  Instead, Ermon,
et al.~show that {\PAWS} behaves like an almost-uniform generator with 
probability greater than $1-\delta$, for a suitable $\delta$ that depends 
on difficult-to-estimate input parameters.  In contrast, we use the estimate
of $|R_F|$ to determine a \emph{small range} of candidate values of $m$.  
This allows us to prove that {\UniGen} is almost-uniform generator with 
confidence $1$.

\begin{algorithm}[htb]
\caption{UniGen$(F, \varepsilon, S)$}
\label{alg:unigen}
\begin{algorithmic}[1]
\Statex \hspace*{-0.6cm}/*Assume $S = \{x_1, \ldots x_{|S|}\}$ is an independent support of $F$, and $\varepsilon > 1.71$ */ 
\State $(\kappa, \mathrm{pivot}) \gets \ComputeKappaAndPivot(\varepsilon)$;
\State $\mathrm{hiThresh} \gets 1 + (1+\kappa)\mathrm{pivot}$;
\State $\mathrm{loThresh} \gets \frac{1}{1+\kappa} \mathrm{pivot}$;
\State $Y \gets \BoundedSAT(F, \mathrm{hiThresh})$;
\If {$\left(|Y| \le \mathrm{hiThresh} \right)$}
	 \State Let $y_1, \ldots y_{|Y|}$ be the elements of $Y$;
 	 \State Choose $j$ at random from $\{1,\ldots |Y|\}$; \Return $y_j$;
\Else
	 \State $C \gets \approxCounter(F, 0.8, 0.8);$ 
	 \State $q \gets \lceil\log C + \log 1.8- \log \mathrm{pivot}\rceil$;
 	 \State $i \gets q-4$;
	 \Repeat
     	\State $i \gets i+1$;
 		\State Choose $h$ at random from $H_{xor}(|S|, i, 3)$;
 		\State Choose $\alpha$ at random from $\{0, 1\}^{i}$;
		\State $Y \gets \BoundedSAT(F \wedge (h(x_1, \ldots x_{|S|}) = \alpha),\mathrm{hiThresh})$;
	\Until {$\left(\mathrm{loThresh} \le |Y| \le \mathrm{hiThresh} \right)$ or ($i = q$)}
	\If {($|Y| > \mathrm{hiThresh}$) or ($|Y| < \mathrm{loThresh}$)}  
		\State \Return $\bot$
	\Else
 		\State Let $y_1, \ldots y_{|Y|}$ be the elements of $Y$;
 		\State Choose $j$ at random from $[|Y|]$ and \Return $y_j$;
	\EndIf
\EndIf
\end{algorithmic}
\end{algorithm}

\begin{algorithm}[htb]
\caption{$\ComputeKappaAndPivot(t\varepsilon)$}
\label{alg:computeKappa}
\begin{algorithmic}
\State Find $\kappa \in [0, 1)$ such that $\varepsilon = (1+\kappa)(2.23+\frac{0.48}{(1-\kappa)^2}) - 1$ ;
\State $\mathrm{pivot} \gets \lceil 3e^{1/2}(1+\frac{1}{\kappa})^2 \rceil$;
\State \Return $(\kappa, \mathrm{pivot})$

\end{algorithmic}
\end{algorithm}

The pseudocode for {\UniGen} is shown in Algorithm~\ref{alg:unigen}.
{\UniGen} takes as inputs a Boolean CNF formula $F$, a tolerance
$\varepsilon$ ($> 1.71$, for teachnical reasons explained in the Appendix)
and a set $S$ of sampling variables.  It either returns a random
witness of $F$ or $\bot$ (indicating failure).  The algorithm assumes
access to a source of random binary numbers, and to two subroutines:
(i) $\BoundedSAT(F, N)$, which, for every $N > 0$, returns
$\min(|R_F|, N)$ distinct witnesses of $F$, and (ii) an approximate
model counter $\approxCounter(F, \varepsilon', 1-\delta')$.

{\UniGen} first computes two quantities, ``$\mathrm{pivot}$'' and
$\kappa$, that represent the expected size of a ``small'' cell and the
tolerance of this size, respectively.  The specific choices of
expressions used to compute $\kappa$ and ``$\mathrm{pivot}$'' in
{\ComputeKappaAndPivot} are motivated by technical reasons explained
in the Appendix.  The values of $\kappa$ and
``$\mathrm{pivot}$'' are used to determine high and low thresholds
(denoted ``$\mathrm{hiThresh}$'' and ``$\mathrm{loThresh}$''
respectively) for the size of each cell.  Lines 5--7 handle the easy
case when $F$ has no more than ``$\mathrm{hiThresh}$'' witnesses.
Otherwise, {\UniGen} invokes $\approxCounter$ to obtain an estimate,
$C$, of $|R_F|$ to within a tolerance of $0.8$ and with a confidence
of $0.8$.  Once again, the specific choices of the tolerance and
confidence parameters used in computing $C$ are motivated by technical
reasons explained in the Appendix.  The estimate $C$ is then
used to determine a range of candidate values for $m$.  Specifically,
this range is $\{q-4, \ldots q\}$, where $q$ is determined in line
10 of the pseudocode.  The loop in lines 12--17 checks whether some
value in this range is good enough for $m$, i.e., whether the number
of witnesses in a cell chosen randomly after partitioning $R_F$ using
$H_{xor}(|S|, m, 3)$, lies within ``$\mathrm{hiThresh}$'' and
``$\mathrm{loThresh}$''.  If so, lines 21--22 return a random witness
from the chosen cell.  Otherwise, the algorithm reports a failure in
line 19.

An probabilistic generator is likely to be invoked multiple times with
the same input constraint in constrained-random verification.  Towards
this end, note than lines 1--11 of the pseudocode need to executed
only once for every formula $F$.  Generating a new random witness
requires executing afresh only lines 12--22.  While this optimization
appears similar to ``leapfrogging''~\cite{SKV13,SKV13MC}, it is
fundamentally different since it does not sacrifice any theoretical
guarantees, unlike ``leapfrogging''.

\noindent {\bfseries \emph{Implementation issues:}} In our
implementation of {\UniGen}, {\BoundedSAT} is implemented using
CryptoMiniSAT~\cite{CryptoMiniSAT} -- a SAT solver that handles xor
clauses efficiently.  CryptoMiniSAT uses \emph{blocking clauses} to
prevent already generated witnesses from being generated again.  Since
the independent support of $F$ determines every satisfying assignment
of $F$, blocking clauses can be restricted to only variables in the
set $S$.  We implemented this optimization in CryptoMiniSAT, leading
to significant improvements in performance.  {\approxCounter} is
implemented using {\approxMC}~\cite{SKV13MC}.  Although the authors
of~\cite{SKV13MC} used ``leapfrogging'' in their experiments, we
disable this optimization since it nullifies the theoretical
guarantees of {~\cite{SKV13MC}.  We use ``random\_device" implemented
in C++ as the source of pseudo-random numbers in lines $7$, $14$, $15$
and $22$ of the pseudocode, and also as the source of random numbers
in {\approxMC}.

\noindent {\bfseries \emph{Guarantees:}} The following
theorem shows that {\UniGen} is an almost-uniform generator with 
a high success probability.
\begin{theorem} \label{theorem:unigen}
If $S$ is an independent support of $F$ and if $\varepsilon > 1.71$,
then for every $y \in R_F$, we have
 \[ \frac{1}{(1+\varepsilon)(|R_F|-1)} \le \prob\left[{\UniGen}(F, \varepsilon, S) = y\right] \le (1+\varepsilon)\frac{1}{|R_F|-1}.\]
In addition, $\prob\left[{\UniGen}(F, \varepsilon, S) \neq \bot\right]
\ge 0.62$.
\end{theorem}
For lack of space, we defer the proof to the Appendix.  It can be
shown that {\UniGen} runs in time polynomial in $\varepsilon^{-1}$ and
in the size of $F$, relative to an {\NP}-oracle.

The guarantees provided by Theorem~\ref{theorem:unigen} are
significantly stronger than those provided by earlier generators that
scale to large problem instances.  Specifically, neither
{\XORSampleprime}~\cite{Gomes-Sampling} nor {\UniWit}~\cite{SKV13}
provide strong upper bounds for the probability of generation of a
witness. {\PAWS}~\cite{EGSS13} offers a \emph{probabilistic} guarantee that
the probability of generation of a witness lies within a tolerance factor of
the uniform probability, while the guarantee of Theorem~\ref{theorem:unigen}
is not prbabilistic.
The success probability of {\PAWS}, like that of {\XORSampleprime}, is 
bounded below by an expression that depends on difficult-to-estimate input 
parameters.  Interestingly, the same parameters also directly affect the 
tolerance of distribution of the generated witnesses.  The success 
probability of {\UniformWitness} is bounded below by $0.125$, which is
significantly smaller than the lower bound of $0.62$ guaranteed by
Theorem~\ref{theorem:unigen}.

\noindent {\bfseries \emph{Trading scalability with uniformity:}} The
tolerance parameter $\varepsilon$ provides a knob to balance
scalability and uniformity in {\UniGen}.  Smaller values of
$\varepsilon$ lead to stronger guarantees of uniformity (by
Theorem~\ref{theorem:unigen}).  Note, however, that the value of
``$\mathrm{hiThresh}$'' increases with decreasing values of
$\varepsilon$, requiring {\BoundedSAT} to find more witnesses.  Thus,
each invocation of {\BoundedSAT} is likely to take longer as
$\varepsilon$ is reduced.

%% file: ExperimentalMethodology.tex
\section{Experimental Results}\label{sec:experiments}
To evaluate the performance of {\UniGen}, we built a prototype
implementation and conducted an extensive set of experiments.
Industrial constrained-random verification problem instances are typically 
proprietary and unavailable for published research.
Therefore, we conducted experiments on CNF SAT constraints
arising from several problems available in the public-domain.  These
included bit-blasted versions of constraints arising in bounded model
checking of circuits and used in~\cite{SKV13}, bit-blasted versions of
SMTLib benchmarks, constraints arising from automated program
synthesis, and constraints arising from ISCAS89 circuits with parity
conditions on randomly chosen subsets of outputs and next-state variables.  

To facilitate running multiple experiments in parallel, we used a 
high-performance cluster and ran each experiment on a node of the cluster.
Each node had two quad-core Intel Xeon processors with 4 GB of main
memory.  Recalling the terminology used in the pseudocode of {\UniGen}
(see Section~\ref{sec:algorithm}), we set the tolerance $\varepsilon$
to $6$, and the sampling set $S$ to an independent support of $F$ in
all our experiments.  Independent supports (not necessarily minimal
ones) for all benchmarks were easily obtained from the providers of
the benchmarks on request.  We used $2,500$ seconds as the timeout for
each invocation of {\BoundedSAT} and $20$ hours as the overall timeout
for {\UniGen}, for each problem instance.  If an invocation of
{\BoundedSAT} timed out in line 16 of the pseudocode of {\UniGen}, we
repeated the execution of lines 14--16 without incrementing $i$.  With
this set-up, {\UniGen} was able to successfully generate random
witnesses for formulas having up to $486,193$ variables.

For performance comparisons, we also implemented and conducted
experiments with {\UniWit} -- a state-of-art near-uniform
generator~\cite{SKV13}.  Our choice of {\UniWit} as a reference for
comparison is motivated by several factors.  First, {\UniGen} and
{\UniWit} share some commonalities, and {\UniGen} can be viewed as an
improvement of {\UniWit}.  Second, {\XORSampleprime} is known to
perform poorly vis-a-vis {\UniWit}~\cite{SKV13}; hence, comparing with
{\XORSampleprime} is not meaningful.  Third, the implementation of
{\PAWS} made available by the authors of~\cite{EGSS13} currently does
not accept CNF formulae as inputs.  It accepts only a graphical model
of a discrete distribution as input, making a direct comparison with
{\UniGen} difficult.  Since {\PAWS} and {\UniWit} share the same
scalability problem related to large random xor-clauses,
we chose to focus only on {\UniWit}.  Since the ``leapfrogging'' heuristic 
used in~\cite{SKV13} nullifies the guarantees of {\UniWit}, we disabled 
this optimization.  For fairness of comparison, we used the same timeouts 
in {\UniWit} as used in {\UniGen}, i.e. $2,500$ seconds for every 
invocation of {\BoundedSAT}, and $20$ hours overall for every invocation of
{\UniWit}.

Table~\ref{table:exptresults} presents the results of our 
performance-comparison experiments.  Column $1$ lists the CNF benchmark, 
and columns $2$ and $3$ give the count of variables and size of
independent support used, respectively.  The results of
experiments with {\UniGen} are presented in the next $3$ columns.
Column $4$ gives the observed probability of success of {\UniGen} when
generating $1,000$ random witnesses. Column $5$ gives the average time
taken by {\UniGen} to generate one witness (averaged over a large
number of runs), while column $6$ gives the average number of
variables per xor-clause used for randomly partitioning $R_F$.  The
next two columns give results of our experiments with {\UniWit}.
Column $7$ lists the average time taken by {\UniWit} to generate a
random witness, and column $8$ gives the average number of variables
per xor-clause used to partition $R_F$.  A ``$-$'' in any column means
that the corresponding experiment failed to generate any witness in
$20$ hours.

It is clear from Table~\ref{table:exptresults} that the average
run-time for generating a random witness by {\UniWit} can be two to
three orders of magnitude larger than the corresponding run-time for
{\UniGen}.  This is attributable to two reasons.  The first stems from
fewer variables in xor-clauses and blocking clauses when small
independent supports are used.  Benchmark ``tutorial3'' exemplifies
this case.  Here, {\UniWit} failed to generate any witness because all
calls to {\BoundedSAT} in {\UniWit}, with xor-clauses and blocking
clauses containing numbers of variables, timed out.  In contrast, the
calls to {\BoundedSAT} in {\UniGen} took much less time, due to short
xor-clauses and blocking clauses using only variables from the
independent support.  The other reason for {\UniGen}'s improved
efficiency is that the computationally expensive step of identifying a
a good range of values for $m$ (see Section~\ref{sec:algorithm} for
details) needs to be executed only once per benchmark.  Subsequently,
whenever a random witness is needed, {\UniGen} simply iterates over
this narrow range of $m$.  In contrast, generating every witness in
{\UniWit} (without leapfrogging) requires sequentially searching over
all values afresh to find a good choice for $m$.  Referring to
Table~\ref{table:exptresults}, {\UniWit} requires more than $20,000$
seconds on average to find a good value for $m$ and generate a random
witness for benchmark ``s953a\_3\_2''.  Unlike in {\UniGen}, there is
no way to amortize this large time over multiple runs in {\UniWit},
while preserving the guarantee of near-uniformity.

Table~\ref{table:exptresults} also shows that the observed success
probability of {\UniGen} is almost always $1$, much higher than what
Theorem~\ref{theorem:unigen} guarantees and better than those from {\UniWit}.    It is clear from our experiments that {\UniGen} can scale
to problems involving almost $500$K variables, while preserving
guarantees of almost uniformity.  This goes much beyond the reach of
any other random-witness generator that gives strong guarantees on the
distribution of witnesses.
\begin{figure}[h!b]\centering
\includegraphics[scale=0.6]{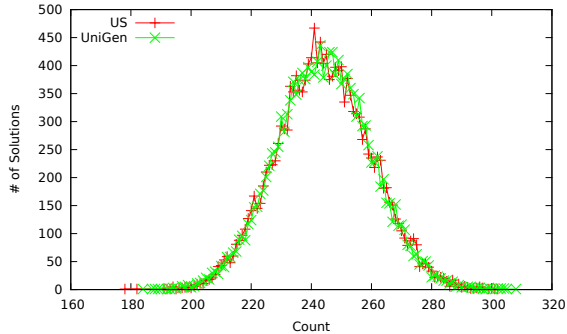}
\vspace*{-0.15in}
\caption{Uniformity comparison for case110}
\label{fig:case110}
\end{figure}

\begin{table*}[t]

\caption{Runtime performance comparison of {\UniGen} and {\UniWit}}
\label{table:exptresults}
\begin{center}
\small{
\begin{tabular}{|c|c|c|c|c|c|c|c|c|}
\hline
\multicolumn{3}{|c|}{}& \multicolumn{3}{|c|}{\UniGen} &
\multicolumn{3}{|c|}{\UniWit} \\
\hline 
Benchmark & |X| & |S| &  \shortstack{Succ\\Prob} & \shortstack{Avg\\Run Time (s)}
&\shortstack{Avg\\XOR leng}  & \shortstack{Avg\\Run Time (s)} & \shortstack{Avg\\XOR len} 
& \shortstack{Succ\\Prob} \\
\hline
Squaring7&1628&72&1.0&2.44&36&2937.5&813&0.87\\\hline
squaring8&1101&72&1.0&1.77&36&5212.19&550&1.0\\\hline
Squaring10&1099&72&1.0&1.83&36&4521.11&550&0.5\\\hline
s1196a\_7\_4&708&32&1.0&6.9&16&833.1&353&0.37\\\hline
s1238a\_7\_4&704&32&1.0&7.26&16&1570.27&352&0.35\\\hline
s953a\_3\_2&515&45&0.99&12.48&23&22414.86&257&*\\\hline
EnqueueSeqSK&16466&42&1.0&32.39&21&--&--&--\\\hline
LoginService2&11511&36&0.98&6.14&18&--&--&--\\\hline
LLReverse&63797&25&1.0&33.92&13&3460.58&31888&0.63\\\hline
Sort&12125&52&0.99&79.44&26&--&--&--\\\hline
Karatsuba&19594&41&1.0&85.64&21&--&--&--\\\hline
tutorial3&486193&31&0.98&782.85&16&--&--&--\\\hline
\end{tabular} 

A ``*'' entry indicates insufficient data for estimating success probability
}
\end{center}
\vspace*{-0.2in}

\end{table*}
\vspace*{-0.2in} Theorem~\ref{theorem:unigen} guarantees that the
probability of generation of every witness lies within a specified
tolerance of the uniform probability.  In practice, however, the
distribution of witnesses generated by {\UniGen} is much more closer
to a uniform distribution.  To illustrate this, we implemented a
\emph{uniform sampler}, henceforth called {\US}, and compared the
distributions of witnesses generated by {\UniGen} and by {\US} for
some representative benchmarks.  Given a CNF formula $F$, {\US} first
determines $|R_F|$ using an exact model counter (such as
{\sharpSATTool}).  To mimic generating a random witness, {\US} simply
generates a random number $i$ in $\{1 \ldots |R_F|\}$.  To ensure fair
comparison, we used the same source of randomness in both {\UniGen}
and {\US}.  For every problem instance on which the comparison was
done, we generated a large number $N$ ($= 4 \times 10^6$) of sample
witnesses using each of {\US} and {\UniGen}.  In each case, the number
of times various witnesses were generated was recorded, yielding a
distribution of the counts.  Figure~\ref{fig:case110} shows the
distributions of counts generated by {\UniGen} and by {\US} for one of
our benchmarks (case110) with $16,384$ witnesses.  The horizontal axis
represents counts and the vertical axis represents the number of
witnesses appearing a specified number of times.  Thus, the point
$(242, 450)$ represents the fact that each of $450$ distinct witnesses
were generated $242$ times in $4\times10^6$ runs.  Observe that the
distributions resulting from {\UniGen} and {\US} can hardly be
distinguished in practice.  This holds not only for this benchmark,
but for all other benchmarks we experimented with.

Overall, our experiments confirm that {\UniGen} is two to three orders
of magnitude more efficient than state-of-the-art random witness
generators, has probability of success almost $1$, and preserves
strong guarantees about the uniformity of generated witnesses.
Furthermore, the distribution of generated witnesses can hardly be
distinguished from that of a uniform sampler in practice.

%% file: Conclusion.tex
\section{Conclusion}\label{sec:conclusion}

Striking a balance between scalability and uniformity is a difficult
challenge when designing random witness generators for
constrained-random verification.  {\UniGen} is the first such
generator for Boolean CNF formulae that scales to hundreds of
thousands of variables and still preserves strong guarantees of
uniformity.  In future, we wish to investigate the design of scalable
generators with similar guarantees for SMT constraints, leveraging recent progress in satisfiability modulo
  theories.

\section*{Acknowledgments}
We profusely thank Mate Soos for implementing required APIs in CryptoMiniSAT without which 
the experimentation section would have been incomplete. Mate was generous with his suggestions to improve our implementation. 
We thank Ajith John for his help in experimental setup.

%% file: Appendix.tex
\setcounter{theorem}{0}
\setcounter{lemma}{0}

In this section, we present a proof of Theorem~\ref{theorem:unigen},
originally stated in Section~\ref{sec:algorithm}, and also present an
extended table of performance comparison results.

Recall that {\UniGen} is a probabilistic algorithm that takes as
inputs a Boolean CNF formula $F$, a tolerance $\varepsilon$ and a
subset $S$ of the support of $F$.  We first show that if $X$ is the
support of $F$, and if $S \subsetneq X$ is an independent support of
$F$, then {\UniGen}($F$, $\varepsilon$, $S$)
behaves \emph{identically} (in a probabilistic sense) to
{\UniGen}($F$, $\varepsilon$, $X$).  Once this is established, the
remainder of the proof proceeds by making the simplifying assumption
$S = X$.

Clearly, the above claim holds trivially if $X = S$.  Therefore, we
focus only on the case when $S \subsetneq X$.  For notational
convenience, we assume $X = \{x_1, \ldots x_n\}$, $0 \le k < n$, $S
= \{x_1, \ldots x_k\}$ and $D = \{x_{k+1}, \ldots x_n\}$ in all the
statements and proofs in this section.  We also use $\vec{X}$ to
denote the vector $(x_1, \ldots x_n)$, and similarly for $\vec{S}$ and
$\vec{D}$.
\begin{lemma}
\label{lemma:ind-decomp}
Let $F(\vec{X})$ be a Boolean function with support $X$, and let $S$
be an independent support of $F$.  Then there exist Boolean functions
$g_0, g_1, \ldots g_{n-k}$, each with support $S$ such that
\[ 
F(\vec{X}) \leftrightarrow \left(g_0(\vec{S}) \wedge \bigwedge_{j=1}^{n-k}(x_{k+j} \leftrightarrow g_j(\vec{S}))\right)
\]
\end{lemma}
\begin{proof}
Since $S$ is an independent support of $F$, we have $D = X \setminus
S$ is a dependent support of $F$.  From the definition of a dependent
support, there exist Boolean functions $g_1, \ldots g_k$, each with
support $S$, such that $F(\vec{X}) \rightarrow \bigwedge_{j=1}^{n-k}
(x_{k+j} \leftrightarrow g_j(\vec{S}))$.

Let $g_0(\vec{S})$ be the characteristic function of the projection
of $R_F$ on $S$.  More formally, $g_0(\vec{S}) \equiv
\bigvee_{(x_{k+1}, \ldots x_n) \in \{0, 1\}^{n-k}} F(\vec{X})$.
It follows that $F(\vec{X}) \rightarrow g_0(\vec{S})$.  Combining this
with the result from the previous paragraph, we get the implication
$F(\vec{X}) \;\rightarrow\;$ $\left(g_0(\vec{S}) \wedge \bigwedge_{j=1}^{n-k}(x_{k+j}\leftrightarrow
g_j(\vec{S}))\right)$

From the definition of $g_0(\vec{S})$ given above, we have
$g_0(\vec{S}) \rightarrow F(\vec{S}, x_{k+1}, \ldots x_n)$, for some
values of $x_{k+1}, \ldots x_n$.  However, we also know that
$F(\vec{X}) \rightarrow \bigwedge_{j=1}^{n-k} (x_{k+j} \leftrightarrow
g_j(\vec{S}))$.  It follows that
$\left(g(\vec{S}) \wedge \bigwedge_{j=1}^{n-k}
(x_{k+j} \leftrightarrow g_j(\vec{S}))\right) \rightarrow F(\vec{X})$.
\end{proof}

Referring to the pseudocode of {\UniGen} in
Section~\ref{sec:algorithm}, we observe that the only steps that
depend directly on $S$ are those in line $14$, where $h$ is chosen
randomly from $H_{xor}(|S|, i, 3)$, and line $16$, where the set $Y$
is computed by calling {\BoundedSAT}($F \wedge (h(x_1, \ldots x_{|S|})
= \alpha), \mathrm{hiThresh}$).  Since all subsequent steps of the
algorithm depend only on $Y$, it suffices to show that if $S$ is an
independent support of $F$, the probability distribution of $Y$
obtained at line $16$ is \emph{identical} to what we would obtain if
$S$ was set equal to the entire support, $X$, of $F$.

The following lemma formalizes the above statement.  As before, we
assume $X = \{x_1, \ldots x_n\}$ and $S = \{x_1, \ldots x_k\}$.
\begin{lemma}
\label{lemma:XequivS}
Let $S$ be an independent support of $F(\vec{X})$.  Let $h$ and
$h'$ be hash functions chosen uniformly at random from $H_{xor}(k, i,
3)$ and $H_{xor}(n, i, 3)$, respectively.  Let $\alpha$ and
$\alpha'$ be tuples chosen uniformly at random from $\{0, 1\}^i$.
Then, for every $Y \in \{0, 1\}^n$ and for every $t > 0$, we have\\
$\prob\left[{\BoundedSAT}\left(F(\vec{X}) \wedge (h(\vec{S}) = \alpha), t\right) = Y\right]$ $= $\\
\hspace*{1in}
$\prob\left[{\BoundedSAT}\left(F(\vec{X}) \wedge (h'(\vec{X}) = \alpha'), t\right) = Y\right]$
\end{lemma}
\begin{proof}
Since $h'$ is chosen uniformly at random from $H_{xor}(n, i, 3)$,
recalling the definition of $H_{xor}(n, i, 3)$, we have
$F(\vec{X}) \wedge (h'(\vec{X}) = \alpha')$ $\equiv$
$F(\vec{X}) \wedge \bigwedge_{l=1}^i \left((a_{l,0} \oplus \bigoplus_{j=1}^{n}
a_{l,j}\cdot x[j]) \leftrightarrow \alpha'[l]\right)$, where the
$a_{l,j}$s are chosen independently and identically randomly from
$\{0, 1\}$.

Since $S$ is an independent support of $F$, from
Lemma~\ref{lemma:ind-decomp}, there exist Boolean functions
$g_1, \ldots g_{n-k}$, each with support $S$, such that
$F(\vec{X}) \rightarrow \bigwedge_{j=1}^{n-k} (x_{k+j} \leftrightarrow
g_j(\vec{S}))$.  Therefore, $F(\vec{X}) \wedge (h'(\vec{X})
= \alpha')$ is semantically equivalent to
$F(\vec{X}) \wedge \bigwedge_{l=1}^i \left((a_{l,0} \oplus \bigoplus_{j=1}^{k}
a_{l,j}\cdot x[j] \oplus B) \leftrightarrow \alpha'[l]\right)$, where
$B \equiv \bigoplus_{j=k+1}^{n} a_{l,j}\cdot g_{j-k}(\vec{S})$.
Rearranging terms, we get $F(\vec{X}) \wedge
\bigwedge_{l=1}^i \left((a_{l,0} \oplus  \bigoplus_{j=1}^{k} a_{l,j}\cdot
 x[j]) \leftrightarrow (\alpha'[l] \oplus B)\right)$.

Since $\alpha'$ is chosen uniformly at random from $\{0,1\}^i$ and
since $B$ is independent of $\alpha'$, it is easy to see that
$\alpha'[l] \oplus B$ is a random binary variable with equal
probability of being $0$ and $1$.  It follows that
$\prob\left[{\BoundedSAT}(F(\vec{X}) \wedge (h'(\vec{X}) = \alpha'),
t) = Y\right]$ $=$
$\prob\left[{\BoundedSAT}(F(\vec{X}) \wedge (h(\vec{S}) = \alpha), t)
= Y\right]$.
\end{proof}

Lemma~\ref{lemma:XequivS} allows us to continue with the remainder of
the proof assuming $S = X$.  It has already been shown
in~\cite{Gomes-Sampling} that $H_{xor}(n, m, 3)$ is a $3$-independent
family of hash functions. We use this fact in a key way in the
remainder of our analysis.  The following result about
Chernoff-Hoeffding bounds, proved in~\cite{SKV13MC}, plays an
important role in our discussion.

\input{DetailedProof}

\input{FullTable}

%% file: DetailedProof.tex
\begin{theorem}\label{theorem:chernoff-hoeffding}
Let $\Gamma$ be the sum of $r$-wise independent random variables, each
of which is confined to the interval $[0, 1]$, and suppose
$\expect[\Gamma] = \mu$.  For $0 < \beta \le 1$, if $2 \le r \le \left\lfloor
\beta^{2}\mu e^{-1/2} \right\rfloor \leq 4$ , then $\prob\left[\,|\Gamma - \mu| \ge
  \beta\mu\,\right] \le e^{-r/2}$.
\end{theorem}

Using notation introduced in Section~\ref{sec:prelims}, let $R_F$
denote the set of witnesses of the Boolean formula $F$.  For
convenience of analysis, we assume that $\log (|R_F|-1)
- \log \mathit{pivot}$ is an integer, where $\mathit{pivot}$ is the
quantity computed by algorithm {\ComputeKappaAndPivot} (see
Section~\ref{sec:algorithm}).  A more careful analysis removes this
assumption by scaling the probabilities by constant factors.  Let us
denote $\log (|R_F|-1) - \log \mathit{pivot}$ by $m$.  The expression
used for computing $\mathit{pivot}$ in algorithm
{\ComputeKappaAndPivot} ensures that $\mathrm{pivot} \ge 17$.
Therefore, if an invocation of {\UniGen} does not return from line $7$
of the pseudocode, then $|R_F| \ge 18$.  Note also that the expression
for computing $\kappa$ in algorithm {\ComputeKappaAndPivot} requires
$\varepsilon \ge 1.71$ in order to ensure that $\kappa \in [0, 1)$ can
always be found.

The following lemma shows that $q$, computed in line $10$ of the
pseudocode, is a good estimator of $m$.
\begin{lemma}
\label{lm:qBounds}
$\prob [ q-3 \le m \le q] \ge 0.8$
\end{lemma}
\begin{proof}
Recall that in line $9$ of the pseudocode, an approximate model
counter is invoked to obtain an estimate, $C$, of $|R_F|$ with
tolerance $0.8$ and confidence $0.8$.  By the definition of
approximate model counting, we have $\prob[\frac{C}{1.8} \le |R_F| \le
(1.8)C] \ge 0.8$. Thus,
$\prob[\log C -
\log (1.8) \le \log |R_F| \le \log C + \log (1.8)] \ge 0.8$. It follows that
$\prob[\log C - \log (1.8) - \log pivot - \log
(\frac{1}{1-1/|R_F|}) \le
\log (|R_F|-1) -\log pivot \le \log C - \log pivot + \log (1.8)-\log 
(\frac{1}{1-1/|R_F|})] \ge 0.8$. Substituting $q = \lceil \log C +\log 1.8 - 
\log pivot \rceil$, $m = \log (|R_F|-1) - \log \mathit{pivot}$, $log (1.8) 
= 0.85$ and $\log (\frac{1}{1-1/|R_F|}) \le 0.12$ (since $|R_F| \ge 18$ on
reaching line $10$ of the pseudocode), we get $\prob [ q-3 \le 
m \le q] \ge 0.8$.
\end{proof}

The next lemma provides a lower bound on the probability
of generation of a witness.  Let $w_{i,y,\alpha}$ denote the probability
$\prob\left[\frac{\mathrm{pivot}} {1+\kappa}\right.$ $\left. \le
|R_{F,h,\alpha}| \leq 1+(1+\kappa)\mathrm{pivot} \mbox{ and } h(y)
= \alpha\right.$ $:$ $\left. h \xleftarrow{R} H_{xor}(n, i,
3)\right]$.  The proof of the lemma also provides a lower
bound on $w_{m,y,\alpha}$.

\begin{lemma}\label{lm:lowerBound}
For every witness $y$ of $F$, $\prob[\textrm{y is output}] \ge \frac{0.8 (1-e^{-1})}{(1.06+\kappa)(|R_F|-1)} $
\end{lemma}
\begin{proof}
If $|R_F| \le 1+(1+\kappa)\mathrm{pivot}$, the lemma holds trivially
(see lines $5$--$7$ of the pseudocode).  Suppose $|R_F| \ge
1+(1+\kappa)\mathrm{pivot}$ and let $U$ denote the event that witness
$y \in R_F$ is output by {\UniGen} on inputs $F$, $\varepsilon$ and
$X$.  Let $p_{i,y}$ denote the probability that we return from line
$17$ for a particular value of $i$ with $y$ in $R_{F,h,\alpha}$, where
$\alpha \in \{0,1\}^{i}$ is the value chosen in line $15$.  Then,
$\prob[U] = \sum_{i=q-3}^{q} \frac{1}{|Y|}p_{i,y} \prod_{j =
q-3}^{i-1} (1-p_{j,y})$, where $Y$ is the set of witnesses returned by
{\BoundedSAT} in line $16$ of the pseudocode. Let $f_m = \prob [
q-3 \le m \le q]$.  From Lemma~\ref{lm:qBounds}, we know that $f_m \ge
0.8$. From the design of the algorithm, we also know that
$\frac{1}{1+\kappa}\mathrm{pivot} \le |Y| \le
1+(1+\kappa)\mathrm{pivot}$. Therefore,
$\prob[U] \ge \frac{1}{1+(1+\kappa)\mathrm{pivot}}\cdot p_{m,y} \cdot
f_m$. The proof is now completed by showing
$p_{m,y} \ge \frac{1}{2^m}(1-e^{-1})$. This gives
$\prob[U] \ge \frac{0.8(1-e^{-1})}{(1+(1+\kappa)pivot)2^{m}} \ge \frac{0.8(1-e^{-1})}{(1.06+\kappa)(|R_F|-1)}$. The
last inequality uses the observation that $ 1/\mathrm{pivot} \leq
0.06$.

  To calculate $p_{m,y}$, we first note that since $y \in R_F$, the
  requirement ``$y\in R_{F,h,\alpha}$" reduces to ``$ y \in
  h^{-1}(\alpha)$". For $\alpha \in \{0,1\}^{n}$, we define
  $w_{m,y,\alpha}$ as $\prob\left[\frac{\mathrm{pivot}}
  {1+\kappa}\right.$ $\left. \le |R_{F,h,\alpha}| \leq
  1+(1+\kappa) \right.$ $\left. pivot \mbox{ and } h(y) = \alpha :
  h \xleftarrow{R} H_{xor}(n, m, 3)\right]$. Therefore, $p_{m,y} $ $=
  \Sigma_{\alpha \in \{0,1\}^m} \left(w_{m,y,\alpha}.2^{-m}\right)$. The proof is now completed by showing that
  $w_{m,y,\alpha} \ge (1-e^{-1})/2^{m}$ for every
  $\alpha \in \{0,1\}^{m}$ and $y \in \{0,1\}^{n}$.

Towards this end, let us 
first fix a random $y$. Now we define an indicator variable $\gamma_{z,
\alpha}$ for every $z \in R_F\setminus\{y\}$ such that $\gamma_{z,\alpha} = 1 $ 
if $h(z) = \alpha$, and $\gamma_{z,\alpha} = 0$ otherwise. Let us fix
$\alpha$ and choose $h$ uniformly at random from $H_{xor}(n,m,3)$. The
random choice of h induces a probability distribution on
$\gamma_{z,\alpha} $ such that $E[\gamma_{z,\alpha}]
= \prob[\gamma_{z,\alpha} = 1] = 2^{-m}$.  Since we have fixed $y$,
and since hash functions chosen from $H_{xor} (n,m,3)$ are $3$-wise
independent, it follows that for every distinct $z_a, z_b \in
R_F \setminus \{y\}$, the random variables $\gamma_{z_a,\alpha},
\gamma_{z_b,\alpha}$ are 2-wise independent.
    Let $\Gamma_{\alpha} = \sum_{z \in R_F \setminus\{y\}} \gamma_{z,\alpha}$ and $\mu_{\alpha} = E[\Gamma_{\alpha}]$. Clearly, $\Gamma_{\alpha} = |R_{F,h,\alpha}| -1$ and $\mu_{\alpha} = \sum_{z \in R_F \setminus \{y\}}$ $ E[\gamma_{z,\alpha}]$ $=\frac{|R_F|-1}{2^m}$. 
 Also, $\prob[\frac{\mathrm{pivot}}{1+\kappa} \le |R_{F,h,\alpha}| \le 1+(1+\kappa) \mathrm{pivot}]$  $= \prob[\frac{\mathrm{pivot}}{1+\kappa} -1  \le |R_{F,h,\alpha}|-1 \le (1+\kappa) \mathrm{pivot}]$ $\ge \prob[\frac{\mathrm{pivot}}{1+\kappa}  \le |R_{F,h,\alpha}|-1 \le (1+\kappa) \mathrm{pivot}]$. Using the expression for $\mathrm{pivot}$, we get $2 \le \lfloor e^{-1/2} (1+1/\epsilon)^{2} \cdot \frac{|R_F|-1}{2^m}\rfloor$.  Therefore using Theorem \ref{theorem:chernoff-hoeffding} and substituting $\mathrm{pivot} = (|R_F|-1)/2^m$, we get $\prob[\frac{\mathrm{pivot}}{1+\kappa}  \le |R_{F,h,\alpha}|-1 \le (1+\kappa) \mathrm{pivot}] \ge 1-e^{-1}$. Therefore, $\prob[\frac{\mathrm{pivot}}{1+\kappa} \le |R_{F,h,\alpha}| \le 1+(1+\kappa) \mathrm{pivot}] \ge 1-e^{-1}$ Since $h$ is chosen at random from $H_{xor}(n,m,3)$, we also have $\prob[h(y) = \alpha] = 1/2^m$. It follows that $w_{m,y,\alpha} \ge (1-e^{-1})/2^m$.      
\end{proof}

The next lemma provides an upper bound of $w_{i, y, \alpha}$ and $p_{i, y}$.
\begin{lemma}\label{lm:iterProof} 
For $i < m$, both $w_{i,y,\alpha}$ and $p_{i,y}$ are bounded above by
  $\frac{1}{|R_F|-1}\frac{1}
  {\left(1-\frac{1+\kappa}{2^{m-i}}\right)^2}$.
\end{lemma}
\begin{proof}
We will use the terminology introduced in the proof of Lemma \ref{lm:lowerBound}.
Clearly, $\mu_{\alpha} = \frac{|R_F|-1}{2^i}$. Since each
$\gamma_{z,\alpha}$ is a $0$-$1$ variable,
$\var\left[\gamma_{z,\alpha}\right] \le
\expect\left[\gamma_{z,\alpha}\right]$.  Therefore, $\sigma^2_{z,
  \alpha}$ $\le \sum_{z \neq y, z \in
  R_F}\expect\left[\gamma_{z,\alpha}\right]$ $ \le \sum_{z \in
  R_F}\expect\left[\gamma_{z, \alpha}\right]$
$=\expect\left[\Gamma_\alpha\right] = 2^{-m}(|R_F|-1)$. So  $\prob[\frac{pivot}{1+\kappa} \le |R_{F,h,\alpha}| \le 1+(1+\kappa)\mathrm{pivot}] \leq \prob[|R_{F,h,\alpha}|-1 \leq (1+\kappa)\mathrm{pivot}]$. From Chebyshev's
inequality, we know that $\prob\left[|\Gamma_\alpha - \mu_{z,\alpha}|
  \ge\right.$ $\left.\kappa\sigma_{z,\alpha} \right] \leq 1/\kappa^2$ for every 
  $\kappa >0$. By choosing $\kappa = (1-\frac{1+\kappa}{2^{m-
  i}})\frac{\mu_{z,\alpha}}{\sigma_{z,\alpha}}$, we have $\prob[|R_{F,h,
  \alpha}|-1 \leq (1+\kappa)\mathrm{pivot}]$ $\leq \prob $ $ \left[| (|R_{F,h,
  \alpha}|-1) - \frac{|R_F|-1}{2^i} |\right.$ $\left.\geq (1-\frac{1+\kappa}
  {2^{m-i}})\frac{|R_F|-1}{2^i} \right]$ $\leq \frac{1}{\left(1-\frac{(1+\kappa)}{2^{m-
  i}}\right)^2} \cdot \frac{2^i}{|R_F|-1}$.  Since $h$ is chosen at random from 
  $H_{xor}(n,m,3)$, we also have $\prob[h(y) = \alpha] = 1/2^i$. It follows 
  that $w_{i,y,\alpha} \leq \frac{1}{|R_F|-1}\frac{1}
  {\left(1-\frac{1+\kappa}{2^{m-i}}\right)^2}$.  The bound for $p_{i,y}$ is easily
obtained by noting that $p_{i,y} = \Sigma_{\alpha \in \{0,1\}^i} \left(w_{i,y,\alpha}.2^{-i}\right)$.
\end{proof}
\begin{lemma}
\label{lm:upperBound}
For every witness $y$ of $F$, $\prob[\textrm{y is output}] \le \frac{1+\kappa}{|R_F|-1} (2.23+ \frac{0.48}{(1-\kappa)^2})$
\end{lemma}
\begin{proof}
We will use the terminology introduced in the proof of Lemma \ref{lm:lowerBound}. $\prob[U] = \sum_{i=q-3}^{q} \frac{1}{|Y|}p_{i,y}$ $ \prod_{j = q-3}^{i} (1-p_{j,y}) \le  \frac{1+\kappa}{\mathrm{pivot}} \sum_{i=q-3}^{q} p_{i,y}$. We can sub-divide the calculation of $\prob[U]$ into three cases based on the range of the values m can take.\\
\textbf{Case 1 :} $q-3 \le m \le q$.\\
Now there are four values that $m$ can take. 
\begin{enumerate}
\item $m = q-3$. We know that $p_{i,y} \le \prob[h(y) = \alpha] = \frac{1}{2^i}$.  $\prob[U | m = q-3] \le \frac{1+\kappa}{\mathrm{pivot}} \cdot \frac{1}{2^{q-3}} \frac{15}{8}$. Substituting the value of $\mathrm{pivot}$ and $m$, we get $\prob[U | m = q-3] \le \frac{15(1+\kappa)}{8(|R_F|-1)}$.

\item $m = q-2$. For $i \in [q-2,q]$ $p_{i,y} \le \prob[h(y) = \alpha] = \frac{1}{2^i}$ Using Lemma \ref{lm:iterProof}, we get $p_{q-3,y} \le \frac{1}{|R_F|-1}\frac{1}
  {\left(1-\frac{1+\kappa}{2}\right)^2}$. Therefore, $\prob[U|m = q-2] \leq      \frac{1+\kappa}{\mathrm{pivot}} \frac{1}{|R_F|-1} (\frac{1}{1-\frac{1+\kappa}{2}}) + \frac{1+\kappa}{\mathrm{pivot}}\frac{1}{2^{q-2}} \frac{7}{4}$. Noting that $\mathrm{pivot} = \frac{|R_F|-1}{2^m} > 10$, $\prob[U|m = q-2] \leq \frac{1+\kappa}{|R_F|-1}(\frac{7}{4} + \frac{0.4}{(1-\kappa)^2})$
\item $m = q-1$. For $i \in [q-1,q]$, $p_{i,y} \le \prob[h(y) = \alpha] = \frac{1}{2^i}$. Using Lemma \ref{lm:iterProof}, we get $p_{q-3,y} + p_{q-2,y} \le \frac{1}{|R_F|-1}\left(\frac{1}{\left(1-\frac{1+\kappa}{2^2}\right)} +   \frac{1}
  {\left(1-\frac{1+\kappa}{2}\right)^2}\right)$. Therefore, $\prob[U|m = q-1] \leq \frac{1+\kappa}{\mathrm{pivot}} \left(\frac{1}{|R_F|-1}\left(\frac{1}{\left(1-\frac{1+\kappa}{2^2}\right)^2} +   \frac{1}
  {\left(1-\frac{1+\kappa}{2}\right)^2}\right) + \frac{1}{2^{q-1}} \frac{3}{2}\right)$. Noting that $\mathrm{pivot} = \frac{|R_F|-1}{2^m} > 10$ and $\kappa \leq 1$, $\prob[U|m = q-1] \leq \frac{1+\kappa}{|R_F|-1} (1.9 + \frac{0.4}{(1-\kappa)^2})$. 
\item $m = q$, $p_{q,y} \leq \prob[h(y) = \alpha] = \frac{1}{2^q}$. Using Lemma \ref{lm:iterProof}, we get $p_{q-3,y} + p_{q-2,y} + p_{q-1,y} \leq \frac{1}{|R_F|-1}\left(\frac{1}{\left(1-\frac{1+\kappa}{2^3}\right)^2}\frac{1}{\left(1-\frac{1+\kappa}{2^2}\right)^2} \right.$ $\left. +   \frac{1}
  {\left(1-\frac{1+\kappa}{2}\right)^2}\right)$. Therefore, $\prob[U|m = q] \leq \frac{1+\kappa}{\mathrm{pivot}} \left(\frac{1}{|R_F|-1} \right.$ $\left.\left(
\frac{1}{\left(1-\frac{1+\kappa}{2^3}\right)^2}  +
   \frac{1}{\left(1-\frac{1+\kappa}{2^2}\right)^2} +   \frac{1}
  {\left(1-\frac{1+\kappa}{2}\right)^2}\right) + 1\right)$. Noting that $\mathrm{pivot} = \frac{|R_F|-1}{2^m} > 10$, $\prob[U|m = q] \leq \frac{1+\kappa}{|R_F|-1} (1.58+\frac{0.4}{(1-\kappa)^2})$. 
\end{enumerate}
$\prob[U | q-3 \leq m \leq q] \leq \max_i (\prob[U|m =i])$. Therefore, $\prob[U| q-3 \leq m \leq q] \leq \prob[U|m=q-1] \leq \frac{1+\kappa}{|R_F|-1}(1.9+\frac{0.4}{(1-\kappa)^2})$. 
\\
\textbf{Case 2 :} $m < q-3$. $\prob[U |  m < q-3] \le \frac{1+\kappa}{\mathrm{pivot}} \cdot \frac{1}{2^{q-3}} \frac{15}{8}$. Substituting the value of $\mathrm{pivot}$ and maximizing $m-q+3$, we get $\prob[U |  m < q-3] \le \frac{15(1+\kappa)}{16(|R_F|-1)}$.\\
\textbf{Case 3 :} $m > q$. Using Lemma \ref{lm:iterProof}, we know that $\prob[U | m >q] \leq \frac{1+\kappa}{|R_F|-1} \frac{2^m}{|R_F|-1}  $ $\sum_{i=q-3}^{q} \frac{1}{1-\frac{1+\kappa}{2^{m-i}}}$. The R.H.S. is maximized when $m = q+1$. Hence $\prob[U | m >q] \leq \frac{1+\kappa}{|R_F|-1}$ $ \frac{2^m}{|R_F|-1}  \sum_{i=q-3}^{q} \frac{1}{1-\frac{1+\kappa}{2^{q+1-i}}}$. Noting that $\mathrm{pivot} = \frac{|R_F|-1}{2^m} > 10$ and expanding 
  the above summation $\prob[U | m > q] \leq \frac{1+\kappa}{|R_F|-1} \frac{1}{10}$ $\left( \frac{1}{(1-\frac{1+\kappa}{2^{4}})^2}+\frac{1}
  {(1-\frac{1+\kappa}{2^{3}})^2}+ \right.$ $\left. \frac{1}{(1-\frac{1+\kappa}
  {2^{2}})^2} + \frac{1}{(1-\frac{1+\kappa}
  {2^{1}})^2}\right).$ Using $\kappa \leq 1$ for the first two summation terms, $\prob[U | m > q] \leq \frac{1+\kappa}{|R_F|-1} \cdot \frac{1}{10} \cdot (7.1 + \frac{4}{(1-\kappa)^2})$

  Summing up all the above cases, $\prob[U] = \prob[U|m < q-3] \times \prob[m < q-3] + \prob[U|q-3 \le m \le  q] \times \prob[q-3 \le m \le  q] +\prob[U|m > q] \times \prob[m > q]$. Using $\prob [m < q-1] \le 0.2$, $\prob[m > q] \le 0.2$ and $\prob[q-3 \le m \le q] \le 1$. Therefore, $\prob[U] \leq \frac{1+\kappa}{|R_F|-1} (2.23 + \frac{0.48}{(1-\kappa)^2})$
 
\end{proof}
Combining Lemma ~\ref{lm:lowerBound} and ~\ref{lm:upperBound}, the following theorem is obtained.
\begin{theorem} 
For every witness $y$ of $F$, if $\varepsilon > 1.71$,\\ 
 \[ \frac{1}{(1+\varepsilon)(|R_F|-1)} \le \prob\left[{\UniGen}(F, \varepsilon, X) = y\right] \le (1+\varepsilon)\frac{1}{|R_F|-1}.\]
\end{theorem}

\begin{proof}
The proof is completed by using Lemmas \ref{lm:lowerBound} and \ref{lm:upperBound} and substituting $(1+\varepsilon) = (1+\kappa)(2.23+\frac{0.48}{(1-\kappa)^2})$. To arrive at the results, we use the inequality  $\frac{1.06+\kappa}{0.8(1-e^{-1})} \le (1+\kappa)(2.23+\frac{0.48}{(1-\kappa)^2})$.   
\end{proof}

\begin{theorem}
Algorithm {\UniGen} succeeds (i.e. does not return $\bot$) with probability at least $0.62$.
\end{theorem}
\begin{proof}
If $|R_F| \le 1+(1+\kappa)\mathrm{pivot}$, the theorem holds trivially. Suppose $|R_F| > 1+(1+\kappa)\mathrm{pivot}$ and let $P_{\mathrm{succ}}$ denote the probability that a run of the algorithm {\UniGen} succeeds. Let 
$p_i,$ such that $(q-3 \le i \le q)$ denote the conditional probability that {\UniGen} 
($F$, $\varepsilon$, $X$) terminates in iteration $i$ of the repeat-until loop (line 11-16) 
with $\frac{\mathrm{pivot}}{1+\kappa} \le |R_{F,h,\alpha}| \le 
1+(1+\kappa)\mathrm{pivot}$, given $|R_F| > 1+(1+\kappa)\mathrm{pivot}$.
Therefore, $P_{\mathrm{succ}} = \sum_{i=q-3}^{q} p_{i} \prod_{j = q-3}^{i} (1-p_{j})$. Let $f_m = \prob [ q-3 \le m \le q]$.  
Therefore, $P_{\mathrm{succ}} \ge p_{m} f_m \ge 0.8p_{m}$. The theorem 
is now proved by using Theorem \ref{theorem:chernoff-hoeffding} to show 
that $p_m \ge 1-e^{-3/2} \ge 0.77$.\\
   For every $y \in \{0,1\}^{n}$ and for every $\alpha \in \{0,1\}^{m}$, define an indicator variable $\nu_{y,\alpha}$ as follows: $\nu_{y,\alpha} = 1$ if $h(y) = \alpha$, and $\nu_{y,\alpha} = 0$ otherwise. Let us fix $\alpha$ and $y$ and choose $h$ uniformly at random from $H_{xor}(n ,m,3)$. The random choice of $h$ induces a probability distribution on $\nu_{y,\alpha}$, such that $\prob[\nu_{y,\alpha} = 1] = \prob[h(y) = \alpha] = 2^{-m}$ and $\expect[\nu_{y,\alpha}] = \prob[\nu_{y,\alpha} = 1] = 2^{-m}$. In addition 3-wise independence of hash functions chosen from $H_{xor}(n ,m,3)$ implies that for every distinct $y_a,y_b,y_c \in R_F$, the random variables $\nu_{y_a,\alpha}, \nu_{y_b,\alpha}$ and $\nu_{y_c,\alpha}$ are 3-wise independent. 
    
Let $\Gamma_\alpha = \sum_{y \in R_F} \nu_{y, \alpha}$ and
$\mu_\alpha = \expect\left[\Gamma_\alpha\right]$.  Clearly,
$\Gamma_\alpha = |R_{F, h, \alpha}|$ and $\mu_\alpha = \sum_{y \in
  R_F} \expect\left[\nu_{y, \alpha}\right] = 2^{-m}|R_F|$.
Since $|R_F| > \mathit{pivot}$ and $i-l > 0$, using the expression for
$\mathit{pivot}$, we get $3 \le \left\lfloor e^{-1/2}(1 +
\frac{1}{\varepsilon})^{-2}\cdot\frac{|R_F|}{2^m}
\right\rfloor$. Therefore, using Theorem
\ref{theorem:chernoff-hoeffding},
$\prob\left[\frac{|R_F|}{2^m}.\left(1-\frac{\kappa}{1+\kappa}\right) \leq
  |R_{F,h,\alpha}|\right.$ $\left.\leq
  (1+\kappa)\frac{|R_F|}{2^m} \right] > 1- e^{-3/2}$.
 Simplifying and noting that $\frac{\kappa}{1+\kappa} <
\kappa$ for all $\kappa > 0$, we obtain
$\prob\left[(1+\kappa)^{-1}\cdot \frac{|R_F|}{2^m} \leq
  |R_{F,h,\alpha}|\right.$ $\left.\leq (1+ \kappa)\cdot \frac{|R_F|}{2^m}
  \right] > 1- e^{-3/2}$. Also, $\frac{\mathrm{pivot}}{1+\kappa} = \frac{1}{1+\kappa}\frac{|R_F|-1}{2^m} \le \frac{|R_F|}{(1+\kappa)2^m}$ and $1+(1+\kappa)\mathrm{pivot} = 1+ \frac{(1+\kappa)(|R_F|-1)}{2^m} \ge \frac{(1+\kappa)|R_F|}{2^m}$. Therefore, $p_m = \prob[\frac{\mathrm{pivot}}{1+\kappa} \le |R_{F,h,\alpha}| \le 1+(1+\kappa)\mathrm{pivot}] \ge$ $\prob\left[(1+\kappa)^{-1}\cdot \frac{|R_F|}{2^m} \right.$ $\left.\leq
  |R_{F,h,\alpha}|\right.$ $\left.\leq (1+ \kappa)\cdot \frac{|R_F|}{2^m}
  \right] \ge 1-e^{-3/2}$. 
\end{proof}

%% file: FullTable.tex
Table ~\ref{table:completeresults} presents an extended version of Table ~\ref{table:exptresults}. 
We observe that {\UniGen}  is two to three orders
of magnitude more efficient than state-of-the-art random witness
generators, has probability of success almost $1$ over a large set of benchmarks arising from different domains.

\begin{table*}[t]
\caption{Extended Table of Runtime performance comparison of {\UniGen} and {\UniWit}}
\label{table:completeresults}
\begin{center}
\begin{tabular}{|c|c|c|c|c|c|c|c|}
\hline
\multicolumn{3}{|c|}{}& \multicolumn{3}{|c|}{\UniGen} &
\multicolumn{2}{|c|}{\UniWit} \\
\hline 
Benchmark & \#Variables & |S| &  \shortstack{Succ\\Prob} & \shortstack{Avg\\Run Time (s)}
&\shortstack{Avg\\XOR len}  & \shortstack{Avg\\Run Time (s)} & \shortstack{Avg\\XOR len} \\
\hline
Case121&291&48&1.0&0.19&24&56.09&145\\\hline
Case1\_b11\_1&340&48&1.0&0.2&24&755.97&170\\\hline
Case2\_b12\_2&827&45&1.0&0.33&22&--&--\\\hline
Case35&400&46&0.99&11.23&23&666.14&199\\\hline
Squaring1&891&72&1.0&0.38&36&--&--\\\hline
Squaring8&1101&72&1.0&1.77&36&5212.19&550\\\hline
Squaring10&1099&72&1.0&1.83&36&4521.11&550\\\hline
Squaring7&1628&72&1.0&2.44&36&2937.5&813\\\hline
Squaring9&1434&72&1.0&4.43&36&4054.42&718\\\hline
Squaring14&1458&72&1.0&24.34&36&2697.42&728\\\hline
Squaring12&1507&72&1.0&31.88&36&3421.83&752\\\hline
Squaring16&1627&72&1.0&41.08&36&2852.17&812\\\hline

s526\_3\_2&365&24&0.98&0.68&12&51.77&181\\\hline
s526a\_3\_2&366&24&1.0&0.97&12&84.04&182\\\hline
s526\_15\_7&452&24&0.99&1.68&12&23.04&225\\\hline
s1196a\_7\_4&708&32&1.0&6.9&16&833.1&353\\\hline
s1196a\_3\_2&690&32&1.0&7.12&16&451.03&345\\\hline
s1238a\_7\_4&704&32&1.0&7.26&16&1570.27&352\\\hline
s1238a\_15\_7&773&32&1.0&7.94&16&136.7&385\\\hline
s1196a\_15\_7&777&32&0.97&8.98&16&133.45&388\\\hline
s1238a\_3\_2&686&32&0.99&10.85&16&1416.28&342\\\hline
s953a\_3\_2&515&45&0.99&12.48&23&22414.86&257\\\hline
TreeMax&24859&19&1.0&0.52&10&49.78&12423\\\hline
LLReverse&63797&25&1.0&33.92&13&3460.58&31888\\\hline
LoginService2&11511&36&0.98&6.14&18&--&--\\\hline
EnqueueSeqSK&16466&42&1.0&32.39&21&--&--\\\hline
ProjectService3&3175&55&1.0&71.74&28&--&--\\\hline
Sort&12125&52&0.99&79.44&26&--&--\\\hline
Karatsuba&19594&41&1.0&85.64&21&--&--\\\hline
ProcessBean&4768&64&0.98&123.52&32&--&--\\\hline

tutorial3\_4\_31&486193&31&0.98&782.85&16&--&--\\\hline
\end{tabular} 
\end{center}
\end{table*}